\pdfoutput=1
\documentclass[sigconf,nonacm]{acmart}
\setcopyright{none}
\usepackage{natbib}
\usepackage{dsfont}
\usepackage{xspace}
\makeatletter
\let\@authorsaddresses\@empty
\makeatother
\usepackage[linesnumbered,lined,ruled]{algorithm2e}
\DontPrintSemicolon
\SetKwFor{RepTimes}{Repeat}{~times:}{end}
\SetKwProg{function}{Function}{:}{}
\SetKwProg{parfor}{parallel for}{\:do}{end}
\usepackage{array} 
\usepackage{graphicx}

\usepackage{pgfplots}
\pgfplotsset{compat=1.17}
\usepackage{mathtools}
\usepackage[shortlabels, inline]{enumitem}
\usepackage{color}
\usepackage{tikz}
\SetKwFor{RepTimes}{Repeat}{~times:}{end}

\usepackage{amsmath}
\usepackage{amsthm}
\usepackage{graphicx,epsfig}
\usepackage{caption}
\usepackage[noend]{algpseudocode}

\usepackage{hyperref}
\usepackage{cleveref}
\hypersetup{
    colorlinks=true,
    linktoc=all,
    allcolors=blue,
}

\SetArgSty{textrm}
\SetKwProg{Fn}{Function}{}{}

\DeclareMathOperator{\poly}{polylog}

\newcommand{\E}{\mathbf{E}}

\numberwithin{equation}{section}

\newtheorem{theorem}{Theorem}[section]
\newtheorem{lemma}{Lemma}[section]

\newtheorem{definition}{Definition}[section]

\newtheorem{corollary}{Corollary}[section]

\newcommand{\whp}{\emph{whp}\xspace}

\newcommand{\bucket}{\mathsf{bucket}}

\begin{document}

\title{High Probability Work Efficient Parallel Algorithms}

\author{Chase Hutton}
\affiliation{
    \institution{University of Maryland}
    \country{United States}
}

\author{Adam Melrod}
\affiliation{
    \institution{Harvard University}
    \country{United States}
}

\begin{abstract}
Randomized parallel algorithms for many fundamental problems achieve optimal
linear work in expectation, but upgrading this guarantee to hold with high
probability (\whp) remains a recurring theoretical challenge.  In this paper,
we address this gap for several core parallel primitives.

First, we present the first parallel semisort algorithm achieving $O(n)$ work
\whp\ and $O(\poly n)$ depth \whp, improving upon the $O(n)$ expected work
bound of Gu et al \cite{GSSB}.  Our analysis introduces new concentration arguments based
on simple tabulation hashing and tail bounds for weighted sums of geometric
random variables.  As a corollary, we obtain an integer sorting algorithm for
keys in $[n]$ matching the same bounds.

Second, we introduce a general framework for boosting randomized parallel
graph algorithms from expected to high-probability linear work.  The framework
applies to \emph{locally extendable} problems---those admitting a
deterministic procedure that extends a solution across a graph cut in work
proportional to the cut size.  We combine this with a \emph{culled balanced
partition} scheme: an iterative culling phase removes a polylogarithmic number
of high-degree vertices, after which the remaining graph admits a balanced
random vertex partition \whp\ via a bounded-differences argument.  Applying
work-inefficient \whp\ subroutines to the small pieces and deterministic
extension across cuts yields overall linear work \whp.  We instantiate this
framework to obtain $O(m)$ work \whp\ and polylogarithmic depth \whp\
algorithms for $(\Delta+1)$-vertex coloring and maximal independent set.
\end{abstract}

\maketitle

\section{Introduction}

Randomized parallel algorithms are a cornerstone of modern algorithm design, enabling simple and elegant solutions for a wide variety of problems. In the work-depth model~\cite{jaja}, the efficiency of a parallel algorithm is measured by its \emph{work} (total number of operations) and its \emph{depth} (longest chain of sequential dependencies). Ideally, the work of a parallel algorithm should match the best sequential time complexity, while the depth should be polylogarithmic, enabling efficient parallel execution. For many fundamental problems, randomization provides an attractive path to achieving these goals: algorithms such as Luby's Maximal Independent Set (MIS)~\cite{luby86} and simple randomized coloring~\cite{johansson99, luby-coloring, hutton-melrod} admit parallel implementations with optimal linear work and polylogarithmic depth.

However, for many of these algorithms, the linear work guarantee holds only \emph{in expectation}. While expected linear work is a useful guarantee, it leaves open the possibility of rare but costly executions. A stronger and more desirable guarantee is that the work is $O(n)$ (or $O(m)$ for graph problems) \emph{with high probability} (\whp), that is, with probability $1 - n^{-\Omega(1)}$. Upgrading from expected to high-probability work bounds is a recurring theoretical challenge: the sources of randomness that make these algorithms simple---random sampling, random hashing often introduce dependencies and variance that are difficult to control tightly across the entire execution.

\paragraph{Semisorting.}
Semisorting is the problem of reordering an array of $n$ records so that records with equal keys are contiguous, without requiring distinct keys to appear in any particular order. As a relaxation of sorting that captures the common need to group records by key, semisort is widely used in
the design and analysis of parallel algorithms \cite{Acar2019BatchDynamicConnectivity,Acar2020BatchDynamicTrees, Anderson2020BatchIncrementalMST,
Blelloch2020RandomizedIncrementalJACM, Blelloch2016RadiusStepping,Dhulipala2019LowLatencyGraphStreaming,dhulipala-gbbs,Dhulipala2020ConnectItPVLDB, Dhulipala2020SemiAsymmetricNVRAMPVLDB, Dong2021EfficientSteppingSPAA,Dong2023ParallelBiconnectivityPPOPP,Liu2022BatchDynamicKCoreSPAA, Shi2021ParallelCliqueCounting,Shi2020ButterflyAPOCS, Shun2020PracticalHypergraphPPOPP,Tseng2019EulerTourTreesALENEX}. Gu, Shun, Sun, and Blelloch~\cite{GSSB} give a top-down parallel semisort with
$O(n)$ expected work and $O(\log n)$ depth \whp. Despite its utility, no prior algorithm achieves $O(n)$ work with high probability.

\paragraph{Graph problems.}
Several classical parallel graph algorithms exhibit a similar gap between expected and high-probability work bounds. Luby's celebrated MIS algorithm~\cite{luby86} can be implemented in $O(m)$ expected work and $O(\log^2 n)$ depth \whp. 
For $(\Delta+1)$-vertex coloring, a folklore randomized algorithm---in which each uncolored vertex independently selects a uniformly random available color and keeps it if no neighbor chose the same color---achieves $O(m)$ expected work and $O(\log n)$ depth \whp. To the best of our knowledge, no prior work establishes high-probability linear work bounds for $(\Delta+1)$-coloring.

\paragraph{Our contributions.}
In this paper, we address the gap between expected and high-probability work bounds for several fundamental parallel primitives. We present the first parallel semisort algorithm that achieves $O(n)$ work \whp\ and $O(\poly n)$ depth \whp. Our algorithm follows the top-down framework of Gu et al.~\cite{GSSB}, but introduces new concentration arguments---leveraging simple tabulation hashing and tail bounds for weighted sums of geometric random variables---to upgrade the work bound from expectation to high probability. As a direct corollary, we obtain an unstable integer sorting algorithm for keys in $[n]$ with $O(n)$ work \whp\ and $O(\poly n)$ depth \whp.

Second, we introduce a framework for boosting randomized parallel graph algorithms from expected to high-probability linear work. The framework applies to problems that admit a natural \emph{deterministic extendability} property: given a solution on a subgraph, one can deterministically extend it across a cut using work proportional to the cut size. We combine this with a \emph{culled balanced partition} scheme that decomposes the graph into small balanced pieces via iterative vertex culling and random partitioning. By applying work-inefficient \whp\ algorithms to the small pieces and the deterministic extender to the cuts, the overall work remains linear \whp. We apply this framework to obtain, to the best of our knowledge, the first $O(m)$ \whp\ work algorithm for $(\Delta+1)$-coloring, with $O(\poly n)$ depth \whp. We additionally give an $O(m)$ \whp\ work, $O(\poly n)$ depth \whp\ algorithm for MIS based on Luby's algorithm. We note that Blelloch et al.~\cite{BFS-greedy-mis} previously achieved this for MIS via a different approach.

\section{Preliminaries}\label{sec:prelim}

\subsection{Model and primitives}
We use the work-depth model with arbitrary forking for analyzing parallel
algorithms~\cite{jaja}.  A set of threads operates on a shared
memory, each supporting standard RAM operations plus a \emph{fork} instruction
that spawns child threads.  Memory may be stack-allocated (private to the
allocating thread) or heap-allocated (shared).  The \emph{work} of an algorithm
is the total number of instructions executed, and the \emph{depth} is the length of the longest chain of dependent instructions.

Our algorithms make use of the following standard primitives, each
operating on an array of $n$ elements:
\begin{itemize}
    \item \textbf{Scan.} Given an array $A$ and an associative operator
    $\oplus$, a scan computes the sequence of all prefix sums of $A$
    under $\oplus$. This requires $O(|A|)$ work (assuming $\oplus$ takes
    $O(1)$ work) and $O(\log |A|)$ depth.

    \item \textbf{Reduce.} Given an array $A$ and an associative binary
    function $f$, a reduction computes $f$ applied across all elements
    of $A$. This requires $O(|A|)$ work (assuming $f$ takes $O(1)$ work)
    and $O(\log |A|)$ depth.

    \item \textbf{Partition.} Given an array $A$ and a predicate $P$, a
    partition reorders $A$ so that all elements satisfying $P$ appear
    before those that do not. Using scans, this requires $O(|A|)$ work
    and $O(\log |A|)$ depth.

    \item \textbf{Filter.} Given an array $A$ and a predicate $P$, a
    filter produces a new array consisting exactly of those $a \in A$
    for which $P(a)$ holds, preserving their original order. This
    requires $O(|A|)$ work and $O(\log |A|)$ depth.

    \item \textbf{Comparison sort.} Given an array $A$ of comparable
    elements, a comparison sort produces a sorted permutation of $A$.
    Using parallel merge sort~\cite{cole-merge-sort}, this requires
    $O(|A| \log |A|)$ work and $O(\log |A|)$ depth.
\end{itemize}

\subsection{Concentration bounds}\label{sec:prelim-concentration}

\subsubsection{Chernoff bounds.}
We recall the standard Chernoff bounds for sums of independent Bernoulli
variables.

\begin{theorem}\label{thm:chernoff}
    Let $X = \sum_{i=1}^n X_i$, where $X_i, i = 1,\dots,n$, are independent random variables valued in $\{0,1\}$. Let $\mu = \E[X]$. Then for $\delta \geq 0$,
    \[\Pr[X \geq (1 + \delta)\mu] \leq \exp\left(-\frac{\delta^2 \mu}{2 + \delta}\right),\]
    and for $\delta \in [0,1]$,
    \[\Pr[X \leq (1 - \delta)\mu] \leq \exp\left(-\frac{\delta^2 \mu}{2}\right).\]
\end{theorem}

\subsubsection{Sums of geometric random variables}
We state two concentration results for sums of i.i.d.\ geometric random variables.

\begin{theorem}\label{thm:geom-concentration}
Let $G_1,\dots,G_r$ be independent with $G_i\sim\mathsf{Ge}(p)$ for some $p\in(0,1]$,
and let $G=\sum_{i=1}^r G_i$. Then for every $\lambda\ge 1$,
\[
\Pr\!\left[G\ge \lambda\cdot\E[G]\right]
\le \exp\!\left(-\frac{(\lambda-1)^2}{2\lambda}\,r\right).
\]
\end{theorem}

\begin{proof}
Let $t:=\lambda \E[G]=\lambda r/p$. Consider $t$ independent Bernoulli$(p)$
trials and let $Y\sim\mathrm{Bin}(t,p)$ be the number of successes. The event
$\{G > t\}$ is equivalent to $\{Y < r\}$. Since $\E[Y]=tp=\lambda r$, we
have
\[
\Pr[G>t]=\Pr\!\left[Y < r\right]
=\Pr\!\left[Y\le (1-\delta)\E[Y]\right],
\]
with $\delta:=1-1/\lambda\in[0,1)$. The Chernoff lower tail
(Theorem~\ref{thm:chernoff}) gives
$\Pr[Y \le (1-\delta)\E[Y]] \le \exp(-\delta^2 \lambda r / 2)
= \exp(-(\lambda-1)^2 r / (2\lambda))$.
\end{proof}

\begin{theorem}\label{thm:weighted-geom}
Let $G_1,\dots,G_r$ be independent with $G_i\sim\mathsf{Ge}(1/2)$.
Let $w_1,\dots,w_r\ge 0$ and define
\[
W_1:=\sum_{i=1}^r w_i,\qquad
W_2:=\sum_{i=1}^r w_i^2,\qquad
W_\infty:=\max_{i\in[r]} w_i.
\]
Then for every $t\ge 0$,
\[
\Pr\!\left[\sum_{i=1}^r w_i G_i \ge 2W_1+t\right]
\le
\exp\!\left(-\min\left\{\frac{t^2}{16W_2},\ \frac{t}{8W_\infty}\right\}\right).
\]
\end{theorem}

\begin{proof}
For $G\sim\mathsf{Ge}(1/2)$ and $\theta<\log 2$,
$\E[e^{\theta G}] = \frac12 e^\theta / (1 - \frac12 e^\theta)$.
Let $\psi(\theta):=\log \E[e^{\theta(G-2)}]=\log \E[e^{\theta G}]-2\theta$.
One can verify that $\psi(\theta) \leq 4\theta^2$ for $\theta \in [0,1/4]$.

Let $Y_i:=w_i(G_i-2)$ and $Y:=\sum_i Y_i$, so that
$\sum_i w_i G_i = 2W_1 + Y$.  For $c \in[0,1/(4W_\infty)]$ we have
$c w_i\le 1/4$, hence
$\E[e^{c Y_i}] = \E[e^{(c w_i)(G_i-2)}] \le \exp(4(c w_i)^2)$.
By independence,
\[
\E[e^{c Y}]
=\prod_i \E[e^{c Y_i}]
\le \exp(4c^2 W_2).
\]
Markov's inequality then yields
$\Pr[Y\ge t] \le \exp(-c t + 4c^2 W_2)$.
Choose $c:=\min\{t/(8W_2),\,1/(4W_\infty)\}$.  If $c=t/(8W_2)$, the exponent
is $-t^2/(16W_2)$.  Otherwise $c=1/(4W_\infty)$ and $t>2W_2/W_\infty$, giving
$-c t + 4c^2 W_2 = -t/(4W_\infty) + W_2/(4W_\infty^2) \le -t/(8W_\infty)$.
\end{proof}

\subsubsection{Bounded differences}

\begin{definition}
    A function $f(x_1,\dots,x_n)$ satisfies the \emph{Lipschitz property} with constants $d_i$, $i\in[n]$, if
    $|f(\mathbf{a}) - f(\mathbf{a}')| \leq d_i$
    whenever $\mathbf{a}$ and $\mathbf{a}'$ differ only in the $i$th coordinate.
\end{definition}

\begin{theorem}[McDiarmid \cite{mcdiarmid89}]\label{thm:mobd}
    If $f$ satisfies the Lipschitz property with constants $d_i$ and $\mathbf{X} = (X_1,\dots,X_n)$ are independent random variables, then
    \[\Pr[\left |{f(\mathbf{X}) - \E f(\mathbf{X})} \right| > t] \leq 2\exp\left(-\frac{2t^2}{\sum_{i \leq n} d_i^2}\right).\]
\end{theorem}

\subsection{Hashing}\label{sec:prelim-hashing}

We map a set of $n$ keys from a universe $U$ into $[m]$.

\subsubsection{Universal hashing}
A family $\mathcal H = \{h : U \to [m]\}$ is \emph{universal} (or \emph{2-universal}) if for every pair of distinct keys $x \neq y \in U$, $\Pr_{h \sim \mathcal{H}}[h(x) = h(y)] \leq 1/m$. Standard constructions (e.g., $h_{a,b}(x) = ((ax + b) \bmod p) \bmod m$ for a prime $p$) use $O(\log |U|)$ bits of randomness and evaluate in $O(1)$ time.

\subsubsection{Simple tabulation hashing}\label{sec:simple-tab}
Fix a constant $c$ and an alphabet $\Sigma$, and view each key as a vector
$x = (x_1, \ldots, x_c) \in \Sigma^c$.  Simple tabulation hashing draws $c$
independent random tables $T_1, \ldots, T_c : \Sigma \to \{0,1\}^w$ (where
$w = \log_2 m$) and defines
\[
h(x) \;=\; T_1[x_1] \oplus T_2[x_2] \oplus \cdots \oplus T_c[x_c].
\]
Evaluation takes $O(c) = O(1)$ table lookups and XORs.

\paragraph{Concentration.}
Despite using limited randomness, simple tabulation exhibits Chernoff-type concentration for hashing into bins.

\begin{theorem}[P\v{a}tra\c{s}cu--Thorup~\cite{tab-hasing}]\label{thm:tab-chernoff}
Consider hashing $n$ balls into $m \ge n^{1-1/(2c)}$ bins using simple tabulation on $c$ characters. Let $q$ be an additional query ball, and let $X_q$ be the number of regular balls that hash to a bin chosen as a function of $h(q)$.
Let $\mu = \E[X_q] = n/m$. Then for any constant $\gamma > 0$:
\begin{align}
\text{for all }\delta \le 1:\quad
\Pr\big[|X_q-\mu|>\delta \mu\big]
&< 2\exp\!\big(-\Omega(\delta^2\mu)\big) + m^{-\gamma}, \label{eq:tab-pt-1}\\
\text{for all }\delta=\Omega(1):\quad
\Pr\big[X_q>(1+\delta)\mu\big]
&< (1+\delta)^{-\Omega((1+\delta)\mu)} + m^{-\gamma}. \label{eq:tab-pt-2}
\end{align}
Moreover, for any $m \le n^{1-1/(2c)}$, with probability at least $1-n^{-\gamma}$, every bin receives $n/m \pm O(\sqrt{(n/m)\log^c n})$ balls.
\end{theorem}

\subsection{The placement problem}\label{sec:placement}

The \emph{placement problem} takes an input array $A$ of $k$ records, each carrying a key in $[m]$.  For each $i \in [m]$, there is a target array $B_i$ whose capacity satisfies $|B_i| \ge \alpha \cdot b_i$, where $b_i := |\{x \in A : \mathrm{key}(x) = i\}|$ and $\alpha \ge 2$ is a constant slack factor.  The goal is to place each record of key $i$ into a distinct slot of $B_i$.

\paragraph{Algorithm.}
Partition $A$ into blocks of size $d$ (the final block has size at most $2d$). In each round, each block selects one currently unplaced record and probes a uniformly random slot in its target $B_i$; if the slot is empty, it inserts the record. This can be done using a compare-and-swap operation.

\begin{lemma}\label{lem:placement-block}
Fix a block of size at most $2d$, and let $R$ denote the number of rounds until all records in the block are placed. Then for every $\lambda \ge 1$,
\[
\Pr\!\left[R > \lambda \cdot \frac{2d}{p}\right]
\le \exp\!\left(-\frac{(\lambda-1)^2}{2\lambda} \cdot 2d\right)
= \exp(-\Omega(\lambda d)).
\]
\end{lemma}

\begin{proof}
At any time, at most $b_i$ slots of $B_i$ are occupied, so a random probe succeeds with probability at least $p := 1 - 1/\alpha \ge 1/2$.  Couple each record's number of attempts by an independent $\mathsf{Ge}(p)$ variable, and apply Theorem~\ref{thm:geom-concentration} to the sum of at most $2d$ such variables.
\end{proof}

\begin{corollary}\label{cor:placement}
With $d = \Theta(\log n)$, the placement algorithm finishes in $O(\log n)$ rounds with probability at least $1 - n^{-\Omega(1)}$. Each round performs $O(k/d)$ probes deterministically, so the total work is $O(k)$ \whp.
\end{corollary}

\subsection{Graph algorithms}\label{sec:prelim-graph}

We describe two parallel graph algorithms on a graph $G = (V, E)$ with $n = |V|$ and $m = |E|$.  Each uses $O(m \log n)$ work \whp\ and $O(\log^2 n)$ depth \whp.  In Section~\ref{sec:whp-graph-algorithms} we will boost both of these algorithms to $O(m)$ work \whp\ while preserving polylogarithmic depth.

\paragraph{Maximal independent set.}
Luby's algorithm~\cite{luby86} proceeds in rounds.  In each round, every
remaining vertex $v$ independently draws a random value $r(v)$, and $v$ joins
the independent set if $r(v)$ is a local maximum among its remaining
neighbors.  All newly added vertices and their neighbors are then removed.
Because each round eliminates a constant fraction of the remaining edges in
expectation, $O(\log n)$ rounds suffice \whp.  Each round scans all remaining
edges once, so the total work is $O(m \log n)$ \whp\ and the depth is
$O(\log^2 n)$ \whp.

\paragraph{$(\Delta+1)$-vertex coloring.}
We use the folklore palette-sampling algorithm analyzed in~\cite{hutton-melrod}. The algorithm operates in rounds.  Let $S$
denote the set of currently uncolored vertices, initially $S = V$.  In each
round, every vertex $v \in S$ independently samples a color $c(v)$ uniformly
at random from its \emph{palette} $\mathcal{P}_v$, the set of colors in
$[\Delta+1]$ not used by any already-colored neighbor of $v$.  Because each
vertex has at least $\deg_S(v) + 1$ colors in its palette (where $\deg_S(v)$
counts uncolored neighbors) and at most $\deg_S(v)$ competitors, each vertex
succeeds with constant probability.  If $v$'s sampled color conflicts with a
neighbor $u \in S$ (i.e., $c(v) = c(u)$), both vertices discard their
choices.  Otherwise $v$ is colored and removed from $S$.  After $O(\log n)$
rounds $S$ is empty \whp.  Each round performs $O(m)$ work (scanning all
remaining edges) and $O(\log n)$ depth, giving $O(m \log n)$ total work \whp\
and $O(\log^2 n)$ depth \whp.
\section{Semisort}

The \emph{semisort} problem asks to permute an array of $n$ records so that
records with equal keys are contiguous.  Prior parallel algorithms for semisort achieve either
$O(n)$ expected work with $O(\log n)$ depth~\cite{GSSB,BFGS20} or $O(n)$ work
\whp\ with $O(n^{3/4})$ depth~\cite{Dong23}.  We give the first algorithm that
achieves $O(n)$ work \whp\ while maintaining polylogarithmic depth.

\subsection{Overview}\label{sec:intro}

Our algorithm follows the top-down semisort framework of Gu--Shun--Sun--Blelloch \cite{GSSB} (hereafter GSSB), which achieves $O(n)$ expected work and $O(\log n)$ depth \whp.

\subsubsection{The GSSB framework.}\label{sec:intro-framework}
Let $A$ be an input array of $n$ records with keys drawn from an arbitrary universe $U$. The GSSB algorithm consists of three main phases.

\paragraph{Phase 1: Sampling and bucketing.}
Each record is independently sampled with probability $p_s = \Theta(1/\log n)$
into a sample array $S$.  The samples are sorted to compute, for each key $x$
present in $A$, an approximate count $\sigma_x$ of the multiplicity of $x$
in $A$.  A key $x$ is declared \emph{heavy} if $\sigma_x \ge \tau$ for a
threshold $\tau = \Theta(\log n)$; the remaining records are \emph{light}.
Heavy and light records are separated into sub-arrays $H$ and $L$.  For the
light records, a hash function $h : U \to [B]$ is used to assign each record to
one of $B = \Theta(n / \log^2 n)$ \emph{light buckets}.

\paragraph{Phase 2: Placement.}
For each heavy key $x$, a destination array $D^H_x$ is allocated whose size is based on the sample count $\sigma_x$.  Similarly, for each light bucket $b \in [B]$, a destination array $D^L_b$ is allocated based on the sampled number of light records mapping to $b$.  Concretely, if $\sigma_x$ (resp.\ $\sigma_b$) denotes the count, the destination array is given size $\alpha f(\sigma_x)$ (resp.\ $\alpha f(\sigma_b)$), where $f$ is a carefully chosen allocation function (defined in~\eqref{eq:f-def}) and $\alpha \ge 2$ is a constant slack factor. Records are then distributed into their respective destination arrays using the randomized placement procedure described in section \ref{sec:placement}.

\paragraph{Phase 3: Local semisorting.}
After placement, each light destination array $D^L_b$ is packed into a
contiguous array $C_b$ of size $m_b$.  The task is now to semisort the records
within each $C_b$.  In the GSSB algorithm, this is accomplished by inserting
the records of $C_b$ into a parallel hash table and then sorting according to the hash using a deterministic radix sort. As the parallel hash table used supports $O(1)$ expected work per operation, the work per bucket is $O(m_b)$ in expectation and the total work is $O(n)$ in expectation.

\subsubsection{Our approach.}\label{sec:intro-contributions} Using standard arguments we show that Phases 1 and 2 of the GSSB algorithm can be performed using $O(n)$ work \whp. The main bulk of our technical contribution is thus concerned with upgrading the work bound for Phase 3 from expectation to high probability. We do this in two steps.

\paragraph{Step 1.} The first step concerns deriving a tight bound on the sizes of the packed light buckets $C_b$. As each distinct record hashed to bucket $b$ has multiplicity at most $O(\log^2 n)$ \whp, it suffices to bound the number of distinct records mapped to $b$. Since we hash at most $n$ distinct records into $B = \Theta(n/\log^2 n)$ buckets, the expected number of distinct records per bucket is $O(\log^2 n)$.  Under a fully random hash function, standard balls-into-bins concentration immediately gives a bound of $O(\log^2 n)$ with high probability. However, truly random hash functions require $\Omega(n \log |U|)$ bits to represent and are
impractical.

We observe that simple tabulation hashing suffices. Simple tabulation hashing, described in section \ref{sec:simple-tab}, views each key as a vector of $c = O(1)$ characters $x = (x_1, \ldots, x_c) \in \Sigma^c$ and computes $h(x) = T_1[x_1] \oplus T_2[x_2] \oplus \cdots \oplus T_c[x_c]$, using $c$ independent random look-up tables $T_i : \Sigma \to \{0,1\}^w$, for a total of $O(|\Sigma| \cdot c)$ random words. Evaluation takes $O(1)$ time. P\v{a}tra\c{s}cu and Thorup~\cite{tab-hasing} show that simple tabulation exhibits Chernoff-type per-bin concentration: when hashing $n$ balls into $m \ge n^{1-1/(2c)}$ bins, the load of any fixed bin satisfies \[\Pr[X > (1+\delta)\mu] < (1+\delta)^{-\Omega((1+\delta)\mu)} + m^{-\gamma}\] for $\delta = \Omega(1)$, where $\mu = n/m$ is the expected load. In our setting $\mu = \Theta(\log^2 n)$ and so for any fixed bucket $b$, applying this tail bound gives failure probability $\exp(-\Omega(\log^2 n)) + B^{-\gamma}$. Then, a union bound over all $B$ buckets gives that no bucket receives more than $O(\log^2 n)$ distinct records with high probability, and consequently $\max_b |C_b| = O(\log^4 n)$ with high probability.

\paragraph{Step 2.} The second step replaces the parallel hash table used by GSSB during the local semisorting phase. We do this as follows. In each light bucket $C_b$ with $m_b$ records (of which at most $D \le m_b$ are distinct), we sample a hash function $g : U \to [m_b^K]$ from a 2-universal family (where $K \ge 3$ is a constant) and sort the records of $C_b$ by their hash values using a deterministic radix sort. If $g$ induces no collision among the $D$ distinct keys, which, by universality and a union bound, happens with probability at least $1 - \binom{D}{2}/m_b^K \ge 1/2$, then the sort produces a valid semisort of $C_b$. A linear-time scan detects collisions; if one is found, we resample $g$ and repeat.

We observe that the number of rehash attempts in bucket $b$ is stochastically dominated by a geometric random variable $X_b \sim \mathsf{Ge}(1/2)$, and the total work within bucket $b$ is $O(m_b \cdot X_b)$.  Importantly, since each attempt uses an independently sampled hash function, the random variables
$\{X_b\}_{b \in [B]}$ are independent across buckets.  The total work of Phase~3 is therefore bounded by
\[
  W = O\left(\sum_{b \in [B]} m_b \cdot X_b\right),
\]
a weighted sum of independent $\mathsf{Ge}(1/2)$ random variables. To bound $W$, we use Theorem \ref{thm:weighted-geom} which provides a concentration inequality for weighted sums of i.i.d geometric random. In particular, we apply Theorem \ref{thm:geom-concentration} with $W_1 = \sum_{b} m_b = |L|$, $W_2 = \sum_b m_b^2$, $W_{\infty} = \max_b m_b$, and $t = |L|$, to get that
\[\Pr[W \geq 3|L|] \leq \exp\left(-\min\left\{\frac{|L|^2}{16W_2},\ \frac{|L|}{8W_\infty}\right\}\right).\] 
Note that by step 1 we have that $W_{\infty} = O(\log^4 n)$ \whp. Furthermore, we can bound $W_2$ as $W_2 \leq W_{\infty} \cdot W_1 = O(|L| \log^4 n)$ \whp. Thus, assuming that $|L| \geq n / \log n$, we get that 
\[\Pr[W \geq 3|L|] \leq \exp\left(-\Omega\left(\frac{|L|}{\log^4 n}\right)\right)
  = n^{-\omega(1)},
\]
yielding $W = O(|L|) = O(n)$ with high probability.



\subsubsection{Results and Organization.} 

We state our main result in the following Theorem.

\begin{theorem}\label{thm:main}
There exists a parallel semisort that, given an array of $n$
records, produces a semisorted output using $O(n)$ work and $O(\poly n)$ depth \whp.
\end{theorem}

The remainder of this section is organized as follows. Section~\ref{sec:semisort} presents the full algorithm. Section~\ref{sec:sampling} establishes the sampling and allocation guarantees that ensure Phases~1 and~2 run in linear work. Section~\ref{sec:light} carries out the bucket-size analysis and the high-probability work bound for Phase~3, which is the heart of the argument. Finally, Section~\ref{sec:putting-together} accounts for the possibility of restarts and completes the proof of Theorem \ref{thm:main}.


\subsection{Algorithm}\label{sec:semisort}

Recall $A$ is an array of $n$ records with keys in a universe $U$. The procedure is given in Algorithm~\ref{alg:semisort}.

\subsubsection{Parameters and notation.}\label{sec:semisort-notation}

Fix constants $K \ge 3$, $\alpha \ge 2$, and $c > 0$.  Let
$p_s = \Theta(1/\log n)$ be the sampling probability and
$\tau = \Theta(\log n)$ the heavy threshold.  Let
\[
B := \lceil \log_2(n/\log^2 n) \rceil
= \Theta\left(\frac{n}{\log^2 n}\right)
\]
be the number of light buckets.
See Table~\ref{table} for a summary of notation. We also recall the allocation function used by the GSSB algorithm. For $s \ge 0$, define
\begin{equation}\label{eq:f-def}
f(s) := \frac{s + c\log n + \sqrt{c^2 \log^2 n + 2sc\log n}}{p_s}.
\end{equation}

\begin{table}[t]
\centering
\begin{tabular}{|l|p{0.74\linewidth}|}
\hline
\textbf{Symbol} & \textbf{Meaning} \\
\hline
$A$, $n$ & input array and number of records \\
$H, L$ & heavy- and light-record sub-arrays \\
\hline
$p_s$ & sampling probability, $\Theta(1/\log n)$ \\
$S$ & array of sampled records \\
$\sigma_x$, $\sigma_b$ & sample count of key $x$ (resp.\ bucket $b$) in $S$ \\
$\tau$ & heavy threshold on $\sigma_x$, $\Theta(\log n)$ \\
\hline
$h$ & simple tabulation hash function $U \to [B]$  \\
$B$ & number of light buckets, $\Theta(n/\log^2 n)$ \\
$\bucket(\cdot)$ & bucket index map $L \to [B]$ via $h$ \\
\hline
$f(\cdot)$, $\alpha$ & allocation function and slack factor \\
$D^H_x$ & heavy destination array for key $x$ \\
$D^L_b$ & light destination array for bucket $b$ \\
$C_b$, $m_b$ & packed light bucket $b$ and its size $|C_b|$ \\
\hline
\end{tabular}
\caption{Notation for Section~\ref{sec:semisort}.}\label{table}
\end{table}
\begin{algorithm}[ht]
\caption{Parallel Semisort}\label{alg:semisort}
\textbf{Input:} Array $A$ of $n$ records with keys in $U$.\\
\textbf{Output:} Array $A'$ in semisorted order.
\begin{enumerate}[1.]
\item \textbf{Sampling.}
Independently sample each record $x\in A$ with probability $p_s$ and store
each sampled $x$ in an array $S$.

\item \textbf{Sample counts.}
Sort $S$ and compute the multiplicity $\sigma_x$ for each key $x$ appearing in $S$.

\item \textbf{Heavy/light partition.}
Let $H:=\{x : \sigma_x\ge \tau\}$ and $L := A\setminus H$.

\item \textbf{Heavy placement.}
\begin{enumerate}[(a)]
\item If $|H|< n / \log n$, comparison-sort $H$ and store as one segment.
\item Otherwise, for each key $x\in H$ allocate $D^H_x$ of size $\alpha f(\sigma_x)$.
      Place each record $x\in H$ into $D^H_x$ by the placement algorithm
      (Section~\ref{sec:placement}) with block size $d=\Theta(\log n)$.
      If placement does not finish within $\Theta(\log n)$ rounds,
      go to Step~1.
\end{enumerate}

\item \textbf{Light placement.}
\begin{enumerate}[(a)]
\item If $|L|<n/\log n$, comparison-sort $L$ and store as one segment.
\item Otherwise, compute a simple tabulation hash $h:U\to [B]$ and set
      $\bucket(x):=h(x)$ for each $x\in L$.
      Compute the bucket-sample counts
      $\sigma_b := |\{y\in S \cap L : \bucket(y)=b\}|$ for each $b\in[B]$,
      allocate $D^L_b$ of size $\alpha f(\sigma_b)$,
      and place each record $x\in L$ into $D^L_{\bucket(x)}$ using the
      placement algorithm with $d=\Theta(\log n)$.
      If placement does not finish within $\Theta(\log n)$ rounds,
      go to Step~1.
\end{enumerate}

\item \textbf{Local semisorting.}
For each light bucket $b$ in parallel: pack $D^L_b$ into a contiguous
array $C_b$ via a filter and let $m_b:=|C_b|$.  Then:
\begin{enumerate}[(a)]
\item Sample a 2-universal hash $g:U\to[m_b^K]$.
\item Radix-sort $C_b$ by hash values.
\item Scan adjacent records: if any pair has equal hash values but
      distinct keys (a collision), go to~(a).
\end{enumerate}

\item \textbf{Pack.}
Pack all segments into a contiguous output array $A'$ using prefix sums
and return $A'$.
\end{enumerate}
\end{algorithm}

\subsubsection{Sampling and Allocation Guarantees}\label{sec:sampling}

In this section we verify that the sampling step produces an accurate approximation of
the input distribution and that the destination arrays allocated in Phase~2 are
simultaneously large enough to hold every key's records and small enough in
total to guarantee $O(n)$ space.

Since each of the $n$ records is included independently with probability $p_s = \Theta(1/\log n)$, we have $|S| \sim \mathrm{Bin}(n, p_s)$, and the Chernoff bound (Theorem~\ref{thm:chernoff}) immediately gives $|S| = \Theta(n/\log n)$ with probability $1 - \exp(-\Omega(n/\log n))$.

The following Lemma shows that if $\widehat{W} \sim \mathrm{Bin}(W, p_s)$ is the observed sample count, then $f(\widehat{W}) \ge W$ with high probability.

\begin{lemma}\label{lem:chernoff-inversion}
Let $W \ge 0$ and $\widehat{W} \sim \mathrm{Bin}(W, p_s)$.  Then
$\Pr[f(\widehat{W}) < W] \le n^{-c}$.
\end{lemma}
\begin{proof}
Write $\mu := W p_s = \E[\widehat{W}]$.  Since $f$ is non-decreasing, the
event $\{f(\widehat{W}) < W\}$ is equivalent to $\{\widehat{W} \le s\}$,
where $s$ is the largest integer with $f(s) < W$.
Set $t := \mu - s > 0$.  The inequality $W > f(s)$ gives
$\mu > s + c\log n + \sqrt{c^2 \log^2 n + 2sc\log n}$, hence
$(t - c\log n)^2 > c^2 \log^2 n + 2sc\log n$, which simplifies to
$t^2 > 2c\log n \cdot \mu$.
The Chernoff lower tail (Theorem~\ref{thm:chernoff}) then yields
\[
\Pr[\widehat{W} \le s] = \Pr[\widehat{W} \le \mu - t]
\le \exp\left(-\frac{t^2}{2\mu}\right) = \exp(-c\log n) = n^{-c}. \qedhere
\]
\end{proof}

Applying this lemma via a union bound over all keys and buckets, we can control the probability that any destination array is under-allocated. Let $N_x$ denote the true multiplicity of key $x$ and $N_b := |\{x \in A : \bucket(x) = b\}|$ the true bucket size. Define the \emph{under-allocation event}
\[
E_{\mathrm{under}}
:=
\bigl(\exists\, x \in A : N_x > f(\sigma_x)\bigr)
\lor
\bigl(\exists\, b \in [B] : N_b > f(\sigma_b)\bigr).
\]
For each key $x$, the sample count $\sigma_x \sim \mathrm{Bin}(N_x, p_s)$, so Lemma~\ref{lem:chernoff-inversion} gives $\Pr[N_x > f(\sigma_x)] \le n^{-c}$; A union bound over at most $n$ keys gives failure probability at most $n^{1-c}$. The same argument applied to each of the $B \le n$ buckets, with $\sigma_b \sim \mathrm{Bin}(N_b, p_s)$, contributes another $n^{1-c}$. Thus we deduce the following Lemma.

\begin{lemma}\label{lem:no-under}
$\Pr[E_{\mathrm{under}}] \le O(n^{1-c})$.
\end{lemma}

It remains to check that the total allocated space is $O(n)$. Let $\mathcal{I} := H \cup [B]$ and define
\begin{equation}\label{eq:total-space}
T := \alpha\left(\sum_{x \in H} f(\sigma_x)
+ \sum_{b \in [B]} f(\sigma_b)\right).
\end{equation}

\begin{lemma}\label{lem:total-space}
$T = \Theta(n)$ with probability $1 - n^{-\Omega(1)}$.
\end{lemma}
\begin{proof}
For all $s \ge 0$, the inequality $\sqrt{c^2\log^2 n + 2sc\log n} \le c\log n + \sqrt{2sc\log n}$ gives
\begin{equation}\label{eq:f-upper}
p_s\, f(s) \le s + 2c\log n + \sqrt{2sc\log n}.
\end{equation}
We bound the heavy and light contributions to $\sum_{i \in \mathcal{I}} f(\sigma_i)$ separately.

For the heavy part, applying~\eqref{eq:f-upper} and using $\sum_{x \in H} \sigma_x \le |S|$, $|H| \le |S|/\tau$, and the Cauchy--Schwarz bound $\sum_{x \in H} \sqrt{\sigma_x} \le \sqrt{|H| \cdot \sum_x \sigma_x} \le |S|/\sqrt{\tau}$, we get $p_s \sum_{x \in H} f(\sigma_x) = O(|S|)$.

For the light part, the same argument with $\sum_{b} \sigma_b \le |S|$ and $\sum_b \sqrt{\sigma_b} \le \sqrt{B|S|}$ yields
\[
p_s \sum_{b \in [B]} f(\sigma_b)
\le |S| + 2c\log n \cdot B + \sqrt{2c\log n \cdot B |S|}.
\]
Substituting $B = \Theta(n/\log^2 n)$ and $|S| = \Theta(n/\log n)$, the right-hand side is $\Theta(n/\log n)$ \whp. Combining and dividing by $p_s = \Theta(1/\log n)$ gives $\sum_{i \in \mathcal{I}} f(\sigma_i) = \Theta(n)$, hence $T = \Theta(n)$ \whp.
\end{proof}

\subsubsection{Light-Bucket Analysis and Local Semisorting}\label{sec:light}

We now turn to the core of the argument: showing that Phase~3, the local semisorting of light buckets, runs in $O(n)$ total work with high probability. The analysis proceeds in two stages. First, we establish that every packed bucket $C_b$ has size at most $O(\log^4 n)$. Then we use this bound, together with the independence coming from the rehashing scheme, to concentrate the total work.

\paragraph{Bucket sizes.}

The size of a packed bucket $C_b$ is bounded by the product of two quantities: the number of distinct keys hashed to $b$, and the maximum multiplicity of any such key. We bound each in turn, starting with multiplicities.

If a key $x$ has $\E[\sigma_x] \ge 2\tau$, the Chernoff bound gives $\Pr[\sigma_x < \tau] \le \exp(-\Omega(\log n)) = n^{-\Omega(1)}$; a union bound over at most $n$ keys shows that \whp\ every key classified as light satisfies $N_x p_s = \E[\sigma_x] < 2\tau$ so $N_x = O(\log^2 n)$.

For the number of distinct keys per bucket, we use the fact that simple tabulation hashing admits Chernoff-type concentration. Let $n_*$ be the number of distinct keys in $L$ and $Y_b$ the number of distinct keys mapped to bucket $b$. We view the initial hashing as throwing $n_*$ balls into $B$ bins using simple tabulation on $c$ characters. The expected load of any bucket is $\mu = n_*/B = O(\log^2 n)$, and the prerequisite $B \ge n_*^{1 - 1/(2c)}$ of Theorem~\ref{thm:tab-chernoff} is satisfied for large $n$. Taking $t = \Theta(\log^2 n)$ so that $\delta := t/\mu - 1 = \Omega(1)$, the tail bound~\eqref{eq:tab-pt-2} gives
\[
\Pr[Y_b > t]
\le (1+\delta)^{-\Omega(t)} + B^{-\gamma}
= \exp(-\Omega(\log^2 n)) + B^{-\gamma}.
\]
A union bound over $B \le n$ buckets, with $\gamma \ge 3$, yields $\max_b Y_b = O(\log^2 n)$ \whp. Combining the two bounds gives:

\begin{lemma}\label{lem:light-bucket-size}
Assume $|L| \ge n/\log n$.  With probability $1 - n^{-\Omega(1)}$, every packed light bucket satisfies $|C_b| = O(\log^4 n)$.
\end{lemma}

\paragraph{Total work of local semisorting.}

We now analyze the rehashing loop in Step~(e)(iv) of Algorithm~\ref{alg:semisort}. Consider a single attempt in a bucket $C_b$ with $m_b$ records, of which $D \le m_b$ are distinct. We hash into a range of size $m_b^K$ using a 2-universal family, so by a union bound the collision probability is at most
\[
\binom{D}{2} \cdot \frac{1}{m_b^K}
\le \frac{m_b^2}{2m_b^K}
= \frac{1}{2m_b^{K-2}}
\le \frac{1}{2}
\]
for $K \ge 3$. Each attempt thus succeeds independently with probability at least $1/2$, so the number of attempts is stochastically dominated by $X_b \sim \mathsf{Ge}(1/2)$. Furthermore, the variables $\{X_b\}_{b \in [B]}$ are independent across buckets.

Each attempt in bucket $b$ performs $O(m_b)$ work (hashing, radix sort, and a linear scan for collisions), so the total work is
\[
W = O\left(\sum_{b \in [B]} m_b\, X_b\right).
\]
This is a weighted sum of independent geometric random variables, which we can concentrate using Theorem~\ref{thm:weighted-geom}. Condition on the event of Lemma~\ref{lem:light-bucket-size}, so $W_1 = \sum_b m_b = |L|$, $W_\infty = \max_b m_b = O(\log^4 n)$, and $W_2 = \sum_b m_b^2 \le W_\infty \cdot W_1 = O(|L|\log^4 n)$. Applying the theorem with $t = |L|$ gives
\[
\Pr\left[\sum_b m_b\, X_b \ge 3|L|\right]
\le \exp\left(-\Omega\left(\frac{|L|}{\log^4 n}\right)\right).
\]
Since $|L| \ge n/\log n$ by assumption, the exponent is $\Omega(n/\log^5 n) = \omega(\log n)$, so the failure probability is $n^{-\omega(1)}$. We record this as:

\begin{lemma}\label{lem:light-internal-work}
Assume $|L| \ge n/\log n$.  The total work of Step~(e)(iv) is $O(|L|)$ with probability $1 - n^{-\omega(1)}$.
\end{lemma}

\subsubsection{Restart Analysis and Overall Work Bound}\label{sec:putting-together}

The algorithm restarts whenever a placement phase exceeds $\Theta(\log n)$ rounds.  We now show that restarts are sufficiently rare, and that the total work summed over all iterations is $O(n)$ \whp.

Let $E_H$ (resp.\ $E_L$) denote the event that heavy (resp.\ light) placement does not finish within $O(\log n)$ rounds. Conditioned on $\neg E_{\mathrm{under}}$, every destination array has capacity at least $\alpha$ times the number of records it must receive, so the slack condition for the placement problem is satisfied. Corollary~\ref{cor:placement} then gives $\Pr[E_H \mid \neg E_{\mathrm{under}}] \le n^{-\Omega(1)}$ (and likewise for $E_L$). Together with Lemma~\ref{lem:no-under}, this yields:

\begin{lemma}\label{lem:restart-probs}
Both
$E_H$ and $E_L$ occur with probability at most $n^{-\Omega(1)}$.
\end{lemma}

We now prove Theorem \ref{thm:main}.

\begin{proof}[Proof of Theorem \ref{thm:main}]

No restarts happen with high probability by Lemma \ref{lem:restart-probs} so it suffices to bound the work and depth assuming this.

\emph{Work.}
Sampling is $O(n)$ work. Sorting $S$ by parallel merge sort costs
$O(|S|\log|S|) = O(n)$ \whp. The heavy/light partition is $O(n)$.
For heavy records: if $|H| < n/\log n$, comparison-sorting $H$ costs
$O(|H|\log|H|) = O(n)$ work; otherwise, on the event $\neg E_H$,
placement runs for $O(\log n)$ rounds with $O(|H|/\log n)$ probes per
round, totaling $O(|H|)$.
For light records: if $|L| < n/\log n$, comparison-sorting $L$ costs
$O(|L|\log|L|) = O(n)$ work; otherwise, on the event $\neg E_L$,
placement costs $O(|L|)$ by the same argument, and
Lemma~\ref{lem:light-internal-work} bounds the local-semisorting work
by $O(|L|)$ \whp.
In both cases the total is $O(|H|+|L|) = O(n)$.
The allocated space is $T = \Theta(n)$ \whp\ by
Lemma~\ref{lem:total-space}, so packing destination arrays and
assembling the final output each cost $O(n)$.
Thus the work is at most $O(n)$ work with probability at
least $1 - n^{-\Omega(1)}$.

\emph{Depth.}
Sampling, partitioning, and the final packing step each have
$O(\log n)$ depth.  Sorting $S$ takes $O(\log n)$ depth.
For heavy records: if $|H| < n/\log n$, comparison sort has
$O(\log n)$ depth; otherwise, placement runs for $O(\log n)$ rounds,
each of $O(\log n)$ depth, contributing $O(\log^2 n)$.
For light records: if $|L| < n/\log n$, comparison sort has
$O(\log n)$ depth; otherwise, placement contributes $O(\log^2 n)$
and local semisorting dominates.  In the latter case, each rehash
attempt on a bucket of size $m_b = O(\log^4 n)$ performs a radix sort
with $K = O(1)$ passes of counting sort at base $m_b$, each of depth
$O(m_b)$, for $O(\log^4 n)$ depth per attempt.
Across all $B$ buckets, the maximum number of attempts is
$\max_b X_b$ where $X_b \sim \mathsf{Ge}(1/2)$ independently; a union
bound gives $\Pr[\max_b X_b > \Omega(\log n)] \le B \cdot 2^{-\Omega(\log n)}
\le n^{1-\Omega(1)}$, so $\max_b X_b = O(\log n)$ \whp.
The depth of local semisorting is therefore
$O(\log^4 n \cdot \log n) = O(\log^5 n)$ \whp, which dominates all
other steps.
\end{proof}

\subsection{Application: integer sorting}
\label{subsec:counting-sort}

As a direct application, our semisort yields an unstable counting sort for
integer keys.

\begin{corollary}\label{cor:counting-sort}
Given an array $A$ of $n$ records with keys in $[n]$, an unstable counting
sort of $A$ can be computed in $O(n)$ work \whp\ and $O(\poly n)$ depth
\whp.
\end{corollary}

\begin{proof}
First, apply Theorem~\ref{thm:main} to $A$, producing an array $A'$
in which records sharing the same key are contiguous.  Next, scan $A'$ to
identify group boundaries (positions where the key changes) and write the size
of each group into a counts array $C[0 : n-1]$.  Finally, compute a
prefix sum over $C$ to obtain output offsets, and copy each group into the
corresponding interval of the output array.  The first step dominates,
giving $O(n)$ work \whp\ and $O(\poly n)$ depth \whp.
\end{proof}
\section{Framework for Graph Algorithms}
\label{sec:whp-graph-algorithms}

Our methods apply to graph problems that exhibit specific deterministic extendability properties. To illustrate this, consider a partition of a graph $H$ into disjoint subgraphs $H_0 \sqcup H_1$. Assuming we possess a solution for the problem on $H_0$, we require a deterministic algorithm---performing work linear in the size of the cut---facilitating the computation of a solution on $H_1$ compatible with $H_0$. Subsequently, a secondary bridging operation back to $H_0$ allows for the computation of a complete solution for $H$. 

The core insight of our approach is that some graph problems admit simple algorithms that run in $O(m \poly n)$ work \whp, alongside admitting these deterministic extenders. With these tools, if we can produce a balanced partition of $G$ into subgraphs $H_1, \dots, H_{\poly n}$, balanced such that $|H_i| = O(m/(\poly n)^2)$, we can inductively apply the work-inefficient \whp solution to the internal pieces and the deterministic extender to the cuts. This strategy ultimately yields an overall \whp linear-work algorithm for the entire graph.

\subsection{The Culled Balanced Partition}

Once the existence of the deterministic extender is established, the problem reduces to producing a balanced partition. Because not all graphs can be partitioned in this manner naively, we introduce a structural relaxation which we call a \emph{culled balanced partition}. This partition consists of a small set of removed vertices $C \subseteq V$ (where $|C| \leq \poly n$) and a balanced partition $\{V_1,\dots,V_k\}$ of the remaining induced subgraph $G - C$. Formally

\begin{definition}
\label{def:culled-balanced-partition}
A culled balanced partition of size $k$ is a tuple $(C, \{V_1,\dots,V_k\})$ where $C \subseteq V$ has size $O(\poly n)$ and $\{V_1,\dots,V_k\}$ is a partition of $V \setminus C$ such that $G[V_i]$ has $O(m/k^2)$ edges for all $i$.
\end{definition}

To construct a culled balanced partition of polylogarithmic size, we rely on the method of bounded differences, which dictates that randomly partitioning a vertex set into $\poly m$ pieces yields a balanced distribution of edges \whp, provided the graph lacks excessively high-degree vertices. Therefore, we utilize an iterative culling process to prune the graph until randomly partitioning is viable.

\begin{itemize}
    \item \textbf{Iterative Culling:} We iteratively identify and remove vertices of high degree, adding them to the culled set $C$. If no such vertices exist, we generate a random partition of the remaining vertices. The degree condition is such that at most $\poly n$ vertices are removed in a single iteration. 
    \item \textbf{Progress Guarantees:} After each round, either the degree condition is satisfied (allowing us to apply the random partitioning to achieve balance \whp), or the total number of edges in the remaining graph decreases by at least half.
    \item \textbf{Termination:} Because the edges halve in unsuccessful rounds, this iteration terminates in at most $\log n$ rounds. 
\end{itemize}

Finally, to compute the overall solution, we process the balanced partition and deterministically extend the solution to the small culled set $C$. Because $C$ is very small in size, the work required to process the culled portion contributes negligibly to the overall complexity, allowing us to conclude with a strictly linear work bound for the entire algorithm.

\subsection{Applications}

\subsubsection{\texorpdfstring{$(\Delta+1)$}{(Delta+1)}-Vertex Coloring}

We obtain a linear-work \whp\ algorithm for $(\Delta+1)$-coloring by combining
the standard palette-sampling coloring routine (as a work-inefficient \whp\
subroutine) with a deterministic extender across cuts.

\paragraph{Deterministic extender.}
Let $H$ be partitioned as $H = H_0 \sqcup H_1$, and suppose $H_0$ is already
properly colored by~$\phi$.  For each vertex $v \in V(H_1)$, define the
residual palette
\[
    \mathcal{C}(v)
    :=
    [\Delta+1] \setminus \{\phi(u) : u \in N_G(v) \cap V(H_0)\}.
\]
Computing these palettes requires scanning each cut edge exactly once, and
hence costs $O(|E(H_0, H_1)|)$ work.  Moreover,
\[
    |\mathcal{C}(v)|
    \ge
    \deg_{H_1}(v) + 1
    \qquad \text{for all } v \in V(H_1),
\]
since at most $\Delta - \deg_{H_1}(v)$ colors are removed from a palette of
size $\Delta + 1$.  Consequently, any list-coloring procedure on~$H_1$ with
palettes $\{\mathcal{C}(v)\}$ yields a proper coloring compatible with~$\phi$.

\paragraph{Algorithm.}
Let $(C,\{V_1, \dots, V_k\})$ be a culled balanced partition of~$G$ with
$k = O(\log n)$.  Write $H_i := G[V_i]$ for each piece.  The
algorithm proceeds as follows.
\begin{enumerate}
    \item For $i = 1, \dots, k$: apply the deterministic extender from the
    already-colored vertices $H_1 \sqcup \cdots \sqcup H_{i-1}$ into~$H_i$ to
    obtain residual palettes, then run the palette-sampling coloring routine
    on~$H_i$ with those palettes.
    \item Extend from $H_1 \sqcup \cdots \sqcup H_k$ to the culled set~$C$
    using the deterministic extender, and color~$C$ with the palette-sampling
    routine.
\end{enumerate}

\paragraph{Work and span.}
Write $m_i := |E(H_i)|$ for the number of internal edges of piece~$i$, and
$c_i := |E(H_1 \sqcup \cdots \sqcup H_{i-1},\, H_i)|$ for the number of cut
edges entering~$H_i$.  Denote by $W_{\mathrm{ps}}(m')$ and
$S_{\mathrm{ps}}(m')$ the work and span of the palette-sampling subroutine on
a subgraph with $m'$~edges; by assumption
$W_{\mathrm{ps}}(m') = O(m' \log n)$ \whp.

Processing piece~$H_i$ costs $O(c_i)$ work for the extender plus
$O(W_{\mathrm{ps}}(m_i))$ work for coloring.  The total work is therefore
\[
    \sum_{i=1}^{k} [O(c_i) + O(W_{\mathrm{ps}}(m_i))]
    + O(W_{\mathrm{ps}}(|C|)).
\]
Each edge of~$G$ is a cut edge for at most one piece, so
$\sum_i c_i \le m$.  Because the partition is balanced,
$m_i = O(m / k^2)$ for every~$i$ \whp, and hence
\[
    \sum_{i=1}^{k} W_{\mathrm{ps}}(m_i)
    =
    \sum_{i=1}^{k} O\left(\frac{m}{k^2} \cdot \log n\right)
    =
    O\left(\frac{m}{k} \cdot \log n\right)
    = O(m),
\]
where the last equality uses $k = \log n$.  The culled set
satisfies $|C| \le \poly n$, so its contribution is
$O(\poly^2 n) = o(m)$.  The total work is thus $O(m)$ \whp.

The pieces are processed in sequence.  The extender for piece~$i$ runs in
$O(\log n)$ span (a parallel scan of the cut edges), and the palette-sampling
routine contributes $S_{\mathrm{ps}}(m_i)$ span.  Hence the overall span is
\begin{multline*}
    \sum_{i=1}^{k} [O(\log n) + S_{\mathrm{ps}}(m_i)]
    + O(\log n) + S_{\mathrm{ps}}(|C|)\\=
    O(k \cdot S_{\mathrm{ps}}(m/k^2)) + O(k \log n).
\end{multline*}
As the palette-sampling routine achieves polylogarithmic span, this simplifies to
$O(\poly n)$.

\subsubsection{Maximal Independent Set}

We obtain a linear-work \whp algorithm for maximal independent set by combining Luby's algorithm and a deterministic extender. Luby's algorithm obtains $O(m \log n)$ work \whp, and once the deterministic extender is described, the algorithm and its analysis is identical to that for vertex coloring.

\paragraph{Deterministic extender.}

Let $H$ be partitioned as $H = H_0 \sqcup H_1$, and suppose $M_0$ is a maximal independent set for $H_0$. In parallel, we delete each vertex $v \in V(H_1)$ possessing a cut edge to an element of $M_0$ from $H_1$. Consequently, computing an MIS $M_1$ of the updated $H_1$ yields an MIS $M := M_0 \cup M_1$ for $H$.

\subsection{Constructing an culled balanced partition via culling}
\label{subsec:whp-partition-construction}

Our construction of a culled balanced partition proceeds in two stages. First, a small set of high-degree vertices is culled to enforce a certain max-degree-to-edge ratio condition; second, once this condition is enforced the remaining vertices are randomly partitioned into $k=\Theta(\poly n)$ buckets. The condition  shows that this produces a culled balanced partition.

\subsubsection{Stage I: culling high-degree vertices}
\label{subsubsec:whp-culling}

We first initialize an empty culled set $C \subseteq V$. The culling stage runs in \emph{phases}. In a phase, let $H$ denote the current graph during culling (initially $H\gets G$). Write $e(H)$ for the number of edges in $H$ and $\Delta(H)$ for the max degree of $H$. Set the removal threshold
\[
\tau(H) := \frac{e(H)}{k^4 \log n}.
\]
If $\Delta(H) > \tau$ holds, then in parallel, remove all vertices $v$ from $H$ with $\deg_H(v)>\tau/2$, and add them to the culled set $C$. If the condition does not hold, we are done with the culling step, and we set $G':=G[V\setminus C]$ be the remaining graph.

\begin{lemma}
\label{lem:whp-phase-progress}
At the end of a phase starting with $M$ edges, either $\Delta(H)\le e(H)/k^4 \log n$ holds, or the number of
edges has halved: $e(H)\le M/2$. In particular, the culling process takes $O(m)$ work and $O(\log n)$ rounds.
\end{lemma}

\begin{proof}
Let $H'$ be the graph at the end of the phase with $M':=e(H')$ and $\Delta':=\Delta(H')$. If the phase ends
with $M'\le M/2$, we are done. Otherwise $M'>M/2$, and vertices violating the initial degree condition have been removed, so
$\Delta'\le \tau/2 = M/(2k^4 \log n)$. Since $M'>M/2$, we have $M<2M'$, and therefore
\[
\Delta' \le \frac{M}{2k^4 \log n} < \frac{2M'}{2k^4 \log n} = \frac{M'}{k^4 \log n}.
\]
We obtain the lemma.
\end{proof}

We also require a bound on the size of $C$.

\begin{lemma}
\label{lem:whp-culled-per-phase}
In any phase, the number of removed vertices is at most $4k^4 \log n$. Thus $|C| = O(\poly n)$, as the number of rounds is $O(\log n)$ and $k = O(\poly n)$.
\end{lemma}

\begin{proof}
Let $r$ be the number of vertices removed. Each removed vertex has degree at least $\tau/2$, so removing all such vertices deletes at least $r\tau/4$ edges (since each edge is adjacent to at most two removed vertices). The number of removed edges cannot exceed the total number of edges, so we conclude that $r\tau/4 \leq m$. This simplifies to $r \leq 4 k^4 \log n$, as desired.
\end{proof}

At the end of culling we obtain $G'=(V',E')$ with $V'=V\setminus C$, $n':=|V'|$, $m':=|E'|$, and
\[
\Delta(G') \le \frac{m'}{k^4 \log n}.
\]

\subsubsection{Stage II: random partitioning of the remaining vertices}
\label{subsubsec:whp-random-partition}

Given $G'=(V',E')$, assign each vertex of $V'$ independently and uniformly to a bucket in $[k]$, producing
a partition $V'=V_1\sqcup\cdots\sqcup V_k$ and output $(\{V_1,\dots,V_k\},C)$. For each $i\in[k]$, define
\[
X_i := e\left(G'[V_i]\right).
\]
By linearity of expectation,
\[
\E[X_i] = \frac{m'}{k^2}.
\]

\begin{lemma}
\label{lem:whp-fixed-bucket-conc}
For each fixed $i\in[k]$ and every $t>0$,
\[
\Pr\left[\left|X_i-\E[X_i]\right|>t\right]
\le 2\exp\left(-\frac{t^2}{\Delta(G')\,m'}\right).
\]
\end{lemma}

\begin{proof}
The variable $X_i$ is a function $f_i$ of the labels $(Z_v)_{v\in V'}$, where $Z_v\in[k]$ is the bucket of $v$. Changing the label of a single vertex $v$ can affect $X_i$ only through edges
incident to $v$, hence $f_i$ is Lipschitz with constant $d_v\le \deg_{G'}(v)$. Applying
Theorem~\ref{thm:mobd} with $\sum_v d_v^2\le \sum_v \deg_{G'}(v)^2 \le \Delta(G')\sum_v \deg_{G'}(v)
=2\Delta(G')m'$ yields
\[
\Pr\left[\left|X_i-\E[X_i]\right|>t\right]
\le 2\exp\left(-\frac{2t^2}{2\Delta(G')m'}\right)
=2\exp\left(-\frac{t^2}{\Delta(G')m'}\right),
\]
as desired.
\end{proof}

Applying this lemma to $G'$ with $t = \Theta(m'/k^2)$, and using that $\Delta(G') \leq m'/(k^4 \log n)$, we obtain
\begin{multline*}
    \Pr\left[X_i = O\left(\frac{m}{k^2}\right)\right] \leq 2\exp\left(-\Omega\left(\frac{k^4 \log n}{m'} \frac{1}{m'} \frac{m'^2}{k^4}\right)\right) \\\leq 2\exp\left(-\Omega(\log n)\right) = \frac{1}{n^{\Omega(1)}}.
\end{multline*}
A union bound then allows us to conclude that $X_i = O(m/k^2)$ for all $i$ \whp. Performing the partition requires linear work and logarithmic span, so we conclude the following theorem.

\begin{theorem}
\label{thm:Balanced culled partition}
For any polylogarithmic $k$, executing Stages I and II produces a culled balanced partition $(C, \{V_1, \dots, V_k\})$ in linear work \whp and polylogarithmic span.
\end{theorem}

\subsubsection{Organizing the graph according to the partition}

Once the culled balanced partition $(C, \{V_1, \dots, V_k\})$ is computed, we
reorganize the graph representation so that each piece $G[V_i]$ and the cut
edges between pieces can be accessed efficiently.  We apply the integer sort
from Corollary~\ref{cor:counting-sort} to the vertex array, keyed on partition
index, so that the vertices of each $V_i$ occupy a contiguous segment.  Next,
for each vertex $v \in V_i$, we sort its adjacency list by the partition index
of the neighbor, again using integer sort.  This groups each adjacency list
into a block of \emph{internal} neighbors (those in $V_i$) and a block of
\emph{cut} neighbors (those in other pieces or in~$C$), with the boundary
between them recorded as a splitting point.  The entire materialization
requires $O(m)$ work \whp\ and $O(\log^5 n)$ depth \whp, dominated by the
integer sorts.

\bibliography{refs.bib}{}
\bibliographystyle{plain}
\end{document}